\documentclass[sigconf,nonacm]{acmart}

\usepackage[utf8]{inputenc}

\def\NOTES{0}

\usepackage{algorithm}

\usepackage[noend]{algpseudocode}

\usepackage{tikz}
\usepackage{tikzsymbols}

\usepackage{wrapfig}
\usepackage{enumitem}
\usepackage{svg}

\usepackage{caption}
\usepackage{subcaption}

\newcommand{\ninj}{\Ninja[1.5]}

\newcommand{\MessageType}[1]{\textsc{#1}}
\newcommand{\MPropose}{\MessageType{propose}}
\newcommand{\MAck}{\MessageType{ack}}
\newcommand{\MAckSig}{\MessageType{sig}}
\newcommand{\MVote}{\MessageType{vote}}

\newcommand{\MCertAck}{\MessageType{CertAck}}
\newcommand{\MCommit}{\MessageType{Commit}}

\newcommand{\xHat}{\widehat{x}}
\newcommand{\tauHat}{\widehat{\tau}}
\newcommand{\sigVote}{\phi_{\mathit{vote}}}
\newcommand{\sigCertAck}{\phi_{\mathit{ca}}}
\newcommand{\vote}{\mathit{vote}}
\newcommand{\nilVote}{\mathit{nil}}

\newcommand{\sigAck}{\phi_{\mathit{ack}}}
\newcommand{\comCert}{cc}
\newcommand{\sigmaHat}{\widehat{\sigma}}

\newcommand{\influential}{influential}
\newcommand{\decisionView}{decision view}

\newcommand{\True}{\mathit{true}}
\newcommand{\False}{\mathit{false}}

\newcommand{\Domain}{\mathcal{V}}
\newcommand{\AllProc}{\Pi}
\newcommand{\ByzProc}{\mathcal{B}}

\newcommand{\Validity}{Validity}
\newcommand{\ExtendedValidity}{Extended validity}
\newcommand{\extendedValidity}{extended validity}
\newcommand{\WeakValidity}{Weak validity}
\newcommand{\weakValidity}{weak validity}
\newcommand{\Liveness}{Liveness}
\newcommand{\liveness}{liveness}
\newcommand{\Consistency}{Consistency}
\newcommand{\consistency}{consistency}

\newcommand{\fDecide}{\textsc{Decide}}
\newcommand{\fleader}{\mathit{leader}}
\newcommand{\fsign}{\mathit{sign}}

\newcommand{\xinput}{x^{\mathit{in}}}
\newcommand{\votes}{\mathit{votes}}

\newcommand{\exec}{\rho}
\newcommand{\IConf}{I}
\newcommand{\simFor}[1]{\overset{#1}{\sim}}
\newcommand{\Suspects}{\mathcal{M}}
\newcommand{\pred}{\mathit{pred}}

\newcommand{\T}{\mathcal{T}}
\newcommand{\Protocol}{\mathcal{P}}

\newcommand{\Schedule}{\mathcal{S}}

\algblock[Block]{BeginBlock}{EndBlock}
\algtext*{BeginBlock}
\algtext*{EndBlock}

\algblockdefx[UponReceive]{UponReceive}{EndHandler}[2]{\textbf{upon receive} $\langle$#1$\rangle$ \textbf{from} #2}{}
\algtext*{EndHandler}  

\algblockdefx[UponRBDeliver]{UponRBDeliver}{EndHandler}[2]{\textbf{upon RB-deliver} $\langle$#1$\rangle$ \textbf{from} #2}{}
\algtext*{EndHandler}  

\algblockdefx[UponURBDeliverIn]{UponURBDeliverIn}{EndHandler}[3]{\textbf{upon URB-deliver} $\langle$#1$\rangle$ \textbf{in} #2 \textbf{from} #3}{}
\algtext*{EndHandler}  

\algblockdefx[Upon]{Upon}{EndHandler}[1]{\textbf{upon} #1}{}
\algtext*{EndHandler}  

\algblockdefx[ForEach]{ForEach}{EndFor}[1]{\textbf{for each} #1 \textbf{do}}{}
\algtext*{EndFor}  

\algblockdefx[Routine]{Routine}{EndRoutine}[1]{\textbf{routine} #1}{}
\algtext*{EndRoutine}  

\algblockdefx[Operation]{Operation}{EndOperation}[2]{\textbf{operation} #1(#2)}{}
\algtext*{EndOperation}  

\algblockdefx[Function]{Function}{EndFunction}[2]{\textbf{function} #1(#2)}{}
\algtext*{EndFunction}  

\algnewcommand{\LeftComment}[1]{\(\triangleright\) #1}
\algnewcommand{\LineComment}[1]{\State \LeftComment{#1}}
\algnewcommand{\LineCommentx}[1]{\Statex \LeftComment{#1}}

\if \NOTES 1
    \newcommand{\isInGrayOut}{0}

    \newcommand{\colored}[2]{\ifthenelse{\isInGrayOut=1}{#2}{{\color{#1}{#2}}}}

    \usepackage[normalem]{ulem}
    \newcommand{\coloredst}[1]{\bgroup\markoverwith{{\color{#1}{\rule[0.5ex]{2pt}{0.4pt}}}}\ULon}

    \newcommand{\atcolor}{blue}
    \newcommand{\pkcolor}{orange}
    \newcommand{\yzcolor}{purple}

    \newcommand{\atrev}[1]{\colored{\atcolor}{#1}}
    \newcommand{\atadd}[1]{\atrev{#1}}
    \newcommand{\atremove}[1]{\coloredst{\atcolor}{#1}}
    \newcommand{\atreplace}[2]{\atremove{#1} \atadd{#2}}

    \newcommand{\pkremove}[1]{\coloredst{\pkcolor}{#1}}

    \newcommand{\yzremove}[1]{\coloredst{\yzcolor}{#1}}

\else

    \newcommand{\atrev}[1]{#1}
    \newcommand{\atadd}[1]{#1}
    \newcommand{\atreplace}[2]{#2}
    \newcommand{\atremove}[1]{}

    \newcommand{\pkremove}[1]{}

    \newcommand{\yzremove}[1]{}
\fi

\begin{document}

\title{Revisiting Optimal Resilience of Fast Byzantine Consensus}
\subtitle{(Extended Version)}

\author{Petr Kuznetsov}
\authornote{The author was supported by TrustShare Innovation Chair.}
\email{petr.kuznetsov@telecom-paris.fr}
\affiliation{%
  \institution{LTCI, T\'el\'ecom Paris, Institut Polytechnique de Paris}
  \country{France}
}

\author{Andrei Tonkikh}
\email{andrei.tonkikh@gmail.com}
\affiliation{%
  \institution{National Research University Higher School of Economics}
  \country{Russia}
}

\author{Yan X Zhang}
\email{yan.x.zhang@sjsu.edu}
\affiliation{%
  \institution{San Jos\'e State University}
  \country{United States}
}

\begin{abstract}
    It is a common belief that Byzantine fault-tolerant solutions for consensus are significantly slower than their crash fault-tolerant counterparts.
    Indeed, in PBFT, the most widely known Byzantine fault-tolerant consensus protocol, it takes three message delays to decide a value, in contrast with just two in Paxos.
    This motivates the search for \emph{fast} Byzantine consensus algorithms that can produce decisions after just two message delays \emph{in the common case}, e.g., 
    under the assumption that the current leader is correct and not suspected by correct processes.
    The (optimal) two-step latency comes with the cost of lower resilience: fast Byzantine consensus requires more processes to tolerate the same number of faults.
    In particular, $5f+1$ processes were claimed to be necessary to tolerate $f$ Byzantine failures.

    In this paper, we present a fast Byzantine consensus algorithm that relies on just $5f-1$ processes. Moreover, we show that $5f-1$ is the tight lower bound, correcting a mistake in the earlier work.
    While the difference of just $2$ processes may appear insignificant for large values of $f$,
    it can be crucial for systems of a smaller scale.
    In particular, for $f=1$, our algorithm requires only $4$ processes, which is optimal for any (not necessarily fast) partially synchronous Byzantine consensus algorithm.
\end{abstract}



\settopmatter{printfolios=true}
\maketitle

\section{Introduction} \label{sec:introduction}

\subsection{Fast Byzantine consensus}

\emph{Consensus}~\cite{pease1980reaching} is by far the most studied problem in distributed computing.
It allows multiple processes to unambiguously agree on a single value.
%
Solving consensus allows one to build a \emph{replicated state machine} by reaching agreement on each next command to be executed.
In other words, it allows a group of processes to act as a single correct process, despite crashes or even malicious  behaviour of some of them.
Having implemented the replicated state machine, one can easily obtain an implementation of any object with a sequential specification~\cite{lamport-smr,smr-tutorial}.
This makes consensus algorithms ubiquitous in practical distributed systems~\cite{chubby, zookeeper, bft-smart, hotstuff, hyperledger-fabric}: instead of implementing each service from scratch, software engineers often prefer to have a single highly optimized and well tested implementation of state machine replication.
Hence, it is crucial for consensus algorithms to be as efficient as possible.
In particular, it is desirable to minimize the degree of redundancy, i.e., the number of processes that execute the consensus protocol.

It has been proven~\cite{pease1980reaching} that in order for a consensus algorithm to work despite the possibility of a malicious adversary taking control over a single participant process (the behaviour we call ``\emph{a Byzantine failure}''), the minimum of $4$ processes is required, assuming that the network is \emph{partially synchronous} (i.e., mostly reliable, but may have periods of instability).
More generally, in order to tolerate $f$ Byzantine failures, the minimum of $3f+1$ processes in total is required.

As complexity of partially synchronous consensus protocols cannot be grasped by the worst-case latency, we typically focus on the \emph{common case} when all correct processes agree on the same correct \emph{leader} process.   
Most Byzantine fault-tolerant consensus algorithms with optimal resilience ($3f+1$ processes) require at least three message delays in order to reach agreement in the common case~\cite{pbft,hotstuff,tendermint}.
In contrast to this, if we assume that the processes may fail only by crashing (i.e., if a corrupted process simply stops taking steps as opposed to actively trying to break the system), many consensus algorithms with optimal resilience ($n=2f+1$ for crash faults only) reach agreement within just two communication steps~\cite{paxos, viewstamped-replication,viewstamped-replication-revisited}.
This gap leads to an extensive line of research~\cite{kursawe2002optimistic, fab-paxos, zyzzyva, revisiting-fast-bft-1, revisiting-fast-bft-2, sbft} towards \emph{fast Byzantine consensus algorithms} -- the class of Byzantine fault-tolerant partially-synchronous consensus algorithms that can reach agreement with the same delay as their crash fault-tolerant counterparts.

Kursawe~\cite{kursawe2002optimistic} was the first to propose a Byzantine consensus algorithm that could decide a value after just two communication steps.
However, the ``optimistic fast path'' of the algorithm works only if there are no failures at all.
Otherwise, the protocols falls back to randomized consensus with larger latency.

Martin and Alvisi~\cite{fab-paxos} proposed the first fast Byzantine consensus algorithm that is able to remain ``fast'' even in presence of Byzantine failures.
The downside of their algorithm compared to classic algorithms such as PBFT~\cite{pbft} is that it requires $5f+1$ processes in order to tolerate $f$ Byzantine failures (as opposed to $3f+1$).
They present a generalized version of their algorithm that requires $3f+2t+1$ processes in order to tolerate $f$ failures and remains fast when the actual number of failures does not exceed $t$ ($t \le f$).
%
%
It is then argued that $3f+2t+1$ is the optimal number of processes for a fast $f$-resilient Byzantine consensus algorithm. 


\subsection{Our contributions}

We spot an oversight in the lower bound proof by Martin and Alvisi~\cite{fab-paxos}. 
As we show in this paper, the lower bound of $3f+2t+1$ processes only applies to a restricted class of algorithms that assume that the processes are split into disjoint sets of \emph{proposers} and \emph{acceptors}.


Surprisingly, if the roles of proposers and acceptors are performed by the same processes,
there exists a fast $f$-resilient Byzantine consensus protocol that requires only $5f-1$ processes.
By adding a \atadd{PBFT-like} ``slow path''~\cite{kursawe2002optimistic, fab-paxos, revisiting-fast-bft-2}
\atreplace{one}{we} can obtain a generalized version of the protocol that requires \atreplace{$n=\max\{3f+2t-1, 3f+1\}$}{$n=3f+2t-1$} processes, tolerates up to $f$ Byzantine failures, and remains fast (terminates in two message delays) as long as the number of failures does not exceed $t$ (for any $t: 1 \le t \le f$).
In particular, to the best of our knowledge, this is the first protocol that is able to remain fast in presence of a single Byzantine failure ($t=1$) while maintaining the optimal resilience ($n=3f+1$).

We show that $n=3f+2t-1$ is the true lower bound for the number of processes required for a fast Byzantine consensus algorithm.
While for large values of $f$ and $t$ the difference of just two processes \atreplace{appears}{may appear} insignificant, it can become crucial in smaller-scale systems.
In particular, to avoid a single point of failure in a system while maintaining the optimal latency ($f=t=1$), the protocol requires only $4$ processes (optimal for any partially synchronous Byzantine consensus protocol), as opposed to $6$ required by previous protocols.

\subsection{Roadmap}

In Section~\ref{sec:preliminaries}, we state our model assumptions and recall the problem of consensus.
We describe our fast Byzantine consensus protocol in Section~\ref{sec:algorithm}.
In Section~\ref{sec:lower-bound}, we discuss the applicability of the previously known lower bound and prove that $3f+2t-1$ is the true lower bound on the number of processes for a fast Byzantine consensus algorithm. 
%
\atreplace{We conclude by discussing related work in Section~\ref{sec:related-work}.}{We discuss related work in Section~\ref{sec:related-work} and provide the details on the implementation of the generalized version of the protocol in Appendix~\ref{app:generalized-version}}.
\section{Preliminaries} \label{sec:preliminaries}

\subsection{Model assumptions}

We consider a set $\AllProc$ of $n$ \emph{processes}, $p_1, \ldots, p_n$.
Every process is assigned with an \emph{algorithm} (deterministic state machine) that it is expected to follow.
A process that \emph{deviates} from its algorithm (we sometimes also say \emph{protocol}), by performing a step that is not prescribed by the algorithm or prematurely stopping taking steps, is called \emph{Byzantine}. 

We assume that in an execution of the algorithm, up to $f$ processes can be Byzantine.
We sometimes also consider a subset of executions in which up to $t\leq f$ processes are Byzantine. 
Non-Byzantine processes are called \emph{correct}.

%

The processes communicate by sending messages across \emph{reliable} (no loss, no duplication, no creation) point-to-point communication channels.  
More precisely, if a correct process sends a message to a correct process, the message is eventually received. 
%
The adversary is not allowed to create messages or modify messages in transit.
The channels are authenticated: the sender of each received message can be unambiguously identified.
\atadd{In the proofs of correctness, for simplicity} we assume that there exists a global clock, not accessible to the processes.  

Every process is assigned with a public/private key pair. %
Every process knows the identifiers and public keys of every other process.
The adversary is computationally bounded so that it is unable to compute private keys of correct processes. 

We assume a \emph{partially synchronous} system\atadd{~\cite{DLS88}}: there exists a known \emph{a priori} bound on message delays $\Delta$ that holds \emph{eventually}: there exists a time after which every message sent by a correct process to a correct process is received within $\Delta$ time units.
This (unknown to the processes) time when the bound starts to hold is called \emph{global stabilization time} (GST).
%
\atrev{We assume that the processes have access to loosely synchronized clocks and,
for simplicity, we neglect the time spent on local computations.}

\subsection{The consensus problem}

Each process $p \in \AllProc$ is assigned with an input value $\xinput_p$.
%
%
%
At most once in any execution, a correct process can \emph{decide} on a value $x$ by triggering the callback $\fDecide(x)$.

Any infinite execution of a consensus protocol must satisfy the following conditions:
%
%
\begin{description}
    \item[\Liveness:] Each correct process must eventually decide on some value;
    \item[\Consistency:] No two correct processes can decide on different values;
    \item[\Validity:] We consider two flavors of this property: \begin{description}
        \item[\WeakValidity:] If all processes are correct and propose the same value, 
         then only this value can be decided on;
        \item[\ExtendedValidity:] If all processes are correct, then only a value proposed by some process can be decided on.
    \end{description}
\end{description}
Note that {\extendedValidity} implies {\weakValidity}, but not vice versa.
\atremove{(i.e., {\extendedValidity} is \emph{strictly stronger} than {\weakValidity})}.
Our algorithm solves consensus with {\extendedValidity},
while our matching lower bound holds even for consensus with {\weakValidity}.
%

\section{Fast Byzantine consensus with optimal resilience} \label{sec:algorithm}


In this section, we present our fast Byzantine consensus algorithm, assuming a system of $n \ge 5f-1$ processes and discuss its \atremove{simple} generalization for $n \ge 3f+2t-1$ processes.


The algorithm proceeds in numbered views.
Each process maintains its current \textit{view number}, 
and each view is associated with a single \textit{leader} process 
by an agreed upon map $\fleader: \mathbb{Z}_{>0} \to \AllProc$, 
$\fleader(v) = p_{(v \text{ mod } n) + 1}$.
When all correct processes have the same current view number $v$, we say that process $\fleader(v)$ is \emph{elected}.

The processes execute a \emph{view synchronization protocol} in the background.
We do not provide explicit implementation for it because any implementation from the literature is sufficient~\cite{pbft,quadratic-round-synchronization,linear-round-synchronization}.

The view synchronization protocol must satisfy the following three properties:
\begin{itemize}
    \item The view number of a correct process is never decreased;

    \item In any infinite execution, \atreplace{some}{a} correct leader is elected an infinite number of times. In other words, at any point in the execution, there is a moment in the future when a correct leader is elected;
    
    \item If a correct leader is elected after GST, no correct process will change its view number for the time period of at least $5\Delta$.
\end{itemize}
Initially, the view number of each process is $1$.
Hence, process $\fleader(1)$ is elected at the beginning of the execution.
If $\fleader(1)$ is correct and the network is synchronous from the beginning of the execution (GST$=0$),
our protocol guarantees that all correct processes decide some value before any process changes its view number.

The first leader begins with sending a {\MPropose} message with its current decision estimate to all processes.
If a process accepts the proposal, it sends an {\MAck} message to every other process.
A process \emph{decides} on the proposed value once it receives {\MAck} messages from \atreplace{a quorum ($n-f$) of processes}{$n - f$ processes}. 
Therefore, as long as the leader is correct and the correct processes do not change their views prematurely, every correct process decides after just two communication steps. 

When correct processes change their views, they engage in the \emph{view change protocol},
helping the newly elected leader to obtain a \emph{safe} value to propose equipped with a \atrev{\emph{progress certificate}---a cryptographic certificate that confirms that the value is safe}. (A value is safe in a view if no other value was or will ever be decided in a smaller view).

Our view change protocol consists of two phases: first, the new leader collects \emph{votes} from processes and makes a decision about which value is safe,
and, second, the leader asks $2f+1$ other processes to confirm with a digital signature that they agree with the leader's decision.
%
%
This second phase, not typical for other consensus protocols, \atreplace{is used to ensure}{ensures} that \atreplace{the certificate sizes do not grow indefinitely
in case of a long period of asynchrony}{the size of the progress certificate is limited}.

Once the view change protocol is completed, the new leader \atreplace{runs the normal case protocol}{proposes a safe value}: it sends a {\MPropose} message to every process and waits for $n-f$ acknowledgments.



Below we describe the \atreplace{normal case protocol}{protocol for proposing values} and the view change protocol in more detail. 

\subsection{Proposing a value} \label{subsec:algorithm-normal-case}

We say that a value $x$ is \emph{safe in a view $v$} if no value other that $x$ can be decided in a view $v'<v$. 

\atreplace{A}{The} view change protocol (Section~\ref{subsec:algorithm-view-change}) provides the new leader with a value $\xHat$ and a \emph{\atadd{progress} certificate} $\sigmaHat$ ensuring that $\xHat$ is safe in the current view $v$.
The \atadd{progress} certificate can be independently verified by any process.
%
\atrev{In the first view ($v=1$), any value is safe and the leader simply proposes its own value ($\xHat = \xinput_{\fleader(1)}$ and $\sigmaHat = \bot$).}

To propose a value \atremove{in the normal case} (illustrated in Figure~\ref{subfig:normal-case-example}), the leader $p$ sends the message $\MPropose(\xHat,v,\sigmaHat,\tauHat)$ to all processes, where $\tauHat=\fsign_p((\MPropose, \xHat,v))$.

When a process receives the proposal for the first time in a given view and ensures that $\sigmaHat$ and $\tauHat$ are valid, 
it sends an {\MAck} message containing the proposed value to every process.
Once a process receives $n-f$ acknowledgments for the same pair $(\xHat, v)$, it decides on the proposed value $\xHat$.

\begin{figure}

    \begin{subfigure}[t]{\linewidth}
        \centering
        \includesvg{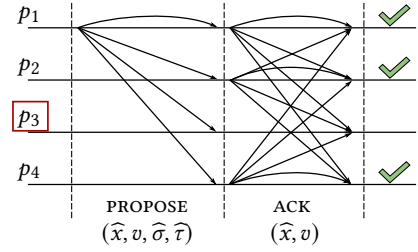}
        \caption{\atrev{Example of a correct process proposing value $\xHat$ in view $v$.
        $\sigmaHat$ is the progress certificate and 
        \mbox{$\tauHat = \fsign_{p_1}((\MPropose, \xHat, v))$}.}}
        \label{subfig:normal-case-example}
    \end{subfigure}
    
    \begin{subfigure}[t]{\linewidth}
        \centering
        \includesvg{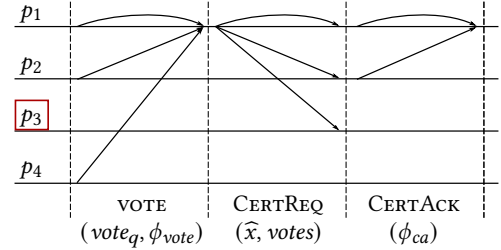}
        \caption{View change execution example. $\sigVote = \fsign_q((\MVote, \vote_q, v))$ and $\sigCertAck = \fsign_q((\MCertAck, \xHat, v))$, where $q$ is the identifier of the process that sends the message.}
        \label{subfig:view-change-example}
    \end{subfigure}
    
    \caption{Execution examples of our protocol.}
    \label{fig:execution-examples}
\end{figure}

\subsection{View change} \label{subsec:algorithm-view-change}

Every process $q$ locally maintains a variable $\vote_q$, an estimate of the value to be decided, in the form $(x,u,\sigma,\tau)$,
where $x$ is a value, $u$ is a view number, $\sigma$ is the progress certificate for value $x$ in view $u$, and $\tau$ is a signature for the tuple $(\MPropose, x,u)$ produced by $\fleader(u)$.
If $\vote_q = (x, u, \sigma, \tau)$, we say that process $q$ votes for ``value $x$ in view $u$''.
Initially, the variable $\vote_q$ has special value $\nilVote$.
When a correct process receives a {\MPropose} message from the leader of its current view for the first time, the process updates its vote by adopting the values from the {\MPropose} message (before sending the {\MAck} message back to the leader).
Note that once a correct \atreplace{replica}{process} changes its vote from $\nilVote$ to something else, it never changes the vote back to $\nilVote$.
We say that a vote is \emph{valid} if either it is equal to $\nilVote$ or both $\sigma$ and $\tau$ are valid with respect to $x$ and $u$.

Whenever a correct process $q$ changes its current view (let $v$ be the new view number), it sends the message $\MVote(\vote_q, \sigVote)$ to the leader of view $v$,
where $\sigVote = \fsign_q((\MVote, \vote_q, v))$.
When a correct \atreplace{replica}{process} finds itself to be the leader of its current view $v$, unless $v=1$, 
it executes the view change protocol (illustrated in Figure~\ref{subfig:view-change-example}).
First, it waits for $n-f$ valid votes and runs the selection algorithm to determine a safe value to propose ($\xHat$).
Then it communicates with other processes to create the \atadd{progress} certificate $\sigmaHat$.

\paragraph{Selection algorithm}

Let $\votes$ be the set of all valid votes received by the leader (with the ids and the signatures of the processes that sent these votes). 
Recall that $|\votes| \ge n-f$.
If all elements in $\votes$ are equal to $\nilVote$, then the leader simply selects its own input value ($\xinput_{\fleader(v)}$).

Otherwise, let $w$ be the highest view number contained in a valid vote.
If there is only one value $x$ such that there is a valid vote $(x,w,*,*)$ in $\votes$, then $x$ is selected.

Let us now consider the case when there are two or more values with valid votes in view $w$.
As a correct leader issues at most one proposal in its view, the only reason for two different valid votes $m_1=(x_1,w,\sigma_1,\tau_1)$ and $m_2=(x_2,w,\sigma_2,\tau_2)$ to exist is that the leader $q$ of view $w$ is Byzantine (we say that process $q$ has \emph{equivocated}).
We can then treat $\gamma=(m_1,m_2)$ as an undeniable evidence of $q$'s misbehavior.
As we have at most $f$ faulty processes, the leader can then wait for $n-f$ votes \emph{not including $q$'s vote}
(i.e., the leader may need to wait for exactly one more vote if $|\votes| = n-f$ and $\votes$ contains a vote from $q$).
After receiving this additional vote, it may happen that $w$ is no longer the highest view number contained in a valid vote.
In this case, the selection algorithm needs to be restarted.

Otherwise, if $w$ remains the highest view number contained in a valid vote, 
let $\votes'$ denote the $n-f$ valid votes from processes other than $q$.
We have two cases to consider:

\begin{enumerate}
    \item[(1)] If there is a set $V \subset \votes'$ of $2f$ valid votes for a value $x$, then $x$ is selected;

    \item[(2)] If no such value $x$ is found, then any value is safe in view $v$. 
    In this case, the leader simply selects its own input value ($\xinput_{\fleader(v)}$).
\end{enumerate}

\paragraph{Creating the progress certificate}

Let $\xHat$ be the value selected by the selection algorithm.
As we prove in Section~\ref{subsec:consistency-proof}, if the leader honestly follows the selection algorithm as described above,
the selected value $\xHat$ will be safe in the current view $v$.
However, the leader also needs to create a certificate $\sigmaHat$ that will prove to all other processes that $\xHat$ is safe.

The naive way to do so is to simply let $\sigmaHat$ be the set of all valid votes received by the leader.
Any process will be able to verify the authenticity of the votes (by checking the digital signatures)
and that the leader followed the selection algorithm correctly (by simulating the selection process locally on the given set of votes).

However, the major problem with this solution is that the certificate sizes will grow without bound in long periods of asynchrony.
%
\atrev{Recall that each vote contains a progress certificate from an earlier view.
If each progress certificate consisted of $n-f$ votes, then each vote would contain a certificate of its own, which, in turn, would consist of $n-f$ votes from an earlier view, and so on.}
If this naive approach is implemented carefully, the progress certificate size (and, hence, the certificate verification time) 
will be linear with respect to the current view number.
While it may be sufficient for some applications (e.g., if long periods of asynchrony are assumed to never happen),
a solution with bounded certificate size would be much more appealing.

In order to bound the \atadd{progress} certificate size, we add an additional round-trip to the view change protocol.
The leader sends the votes alongside the selected value $\xHat$ to at least $2f+1$ different processes and waits for $f+1$ signed
confirmations.
The certificate $\sigmaHat$ is the set of $f+1$ signatures from different \atreplace{replicas}{processes} for the tuple $(\MCertAck, \xHat, v)$.
Intuitively, since there are at most $f$ Byzantine processes in total, it is sufficient to present signatures from $f+1$ \atreplace{replicas}{processes}
to prove that at least one correct \atreplace{replica}{process} verified that the leader performed the selection algorithm correctly
and, hence, that $\xHat$ is safe in view $v$.
As a result, the size of a protocol message does not depend on the view number.







\subsection{Correctness proof} \label{subsec:consistency-proof}

It is easy to see that the protocol satisfies the {\liveness} property of consensus:
once a correct leader is elected after GST, there is nothing to stop it from driving the protocol to \atreplace{termination}{completion}.
The {\extendedValidity} property is immediate.
Hence, in this section, we focus on {\consistency}.
We show that a correct leader always chooses a safe value in the view change protocol.
%
%


Our proofs are based on the following three quorum intersection properties (recall that $n \ge 5f-1$):
\begin{enumerate}[label=(QI\arabic*)]
    \item \label{prop:simple-quorum-intersection}
    \textbf{Simple quorum intersection:} 
    any two sets of $n-f$ processes intersect in at least one correct process.
    This follows from the pigeonhole principle. 
    It is sufficient to verify that $2(n-f) - n \ge f + 1$,
    which is equivalent to $n \ge 3f + 1$ and holds for any $n \ge 5f-1$ assuming that $f \ge 1$;
    
    \item \label{prop:quorum-intersection:n-f:n-f}
    \textbf{Equivocation quorum intersection \#1:} 
    if $Q_1 \subset \AllProc$ such that $|Q_1| = n-f$ and
    $Q_2 \subset \AllProc$ such that $|Q_2| = n-f$ and there are at most $f-1$ Byzantine processes in $Q_2$,
    then $Q_1 \cap Q_2$ contains at least $2f$ correct processes.
    Again, by the pigeonhole principle, it is sufficient to verify that $2(n-f) - n \ge (f-1) + 2f$,
    which is equivalent to $n \ge 5f-1$;
    
    \item \label{prop:quorum-intersection:n-f:2f}
    \textbf{Equivocation quorum intersection \#2:}
    if $Q_1 \subset \AllProc$ such that $|Q_1| = n-f$ and
    $Q_2 \subset \AllProc$ such that $|Q_2| = 2f$ and there are at most $f-1$ Byzantine processes in $Q_2$,
    then $Q_1 \cap Q_2$ contains at least one correct process.
    It is sufficient to verify that $(n-f) + 2f - n \ge (f-1) + 1$, which holds for any \atreplace{values of $n$ and $f$, $n\ge 2f$}{$n \ge 2f$}.
\end{enumerate}



Recall that a value \emph{$x$ is safe in a view $v$} if no value other than $x$ can be decided in a view $v'<v$.
%
To prove that no two different values can be decided in our algorithm, we show that if $x$ and $v$ are equipped with a valid certificate $\sigma$, then $x$ is safe in $v$.

But let us first address the corner case when the leader of a view receives no valid votes other than $\nilVote$.

\begin{lemma} \label{lem:safe-if-all-nil}
If the leader of view $v$ receives $\nilVote$ from $n-f$ different processes during the view change, then any value is safe in $v$.
\end{lemma}
\begin{proof}
Suppose, by contradiction, that at some point of the execution some value $y$ is decided in a view $w'$ smaller than $v$.
Consider the set $Q_1 \subset \Pi$ of $n-f$ processes that acknowledged value $y$ in $w'$.
Consider also the set $Q_2 \subset \Pi$ of $n-f$ processes that sent $\nilVote$ to the leader of view $v$.
By property~\ref{prop:simple-quorum-intersection},
$Q_1 \cap Q_2$ contains at least one correct process.

A correct process only sends messages associated with its current view and it never decreases its current view number.
Hence, it cannot send the vote in view $v$ before sending the acknowledgment in view $w'$.
If the correct process acknowledged value $y$ in $w'$ before sending the vote to the leader of view $v$,
the vote would have not been $\nilVote$---a contradiction.
\end{proof}

We now proceed by induction on view $v$. 
The base case ($v=1$) is immediate: by convention, any value is safe in view $1$.

Now consider a view $v>1$ and assume that for all views $u<v$, any value equipped with a valid certificate for view $u$ is safe in $u$.
Let $w$ denote the highest view number contained in a valid vote received by the leader of view $v$ during the view change protocol.

\begin{lemma} \label{lem:no-decisions-between-views}
No value was or will ever be decided in any view $w'$ such that $w < w' < v$.
\end{lemma}
\begin{proof}
Suppose, by contradiction, that at some point of the execution some other value $y$ is decided in $w' (w < w' < v)$. 
Let $Q_1 \subset \Pi$ be the set of $n-f$ processes that acknowledged value $y$ in $w'$
and let $Q_2 \subset \Pi$ be the set of $n-f$ processes that sent their votes to the leader of view $v$.
By property~\ref{prop:simple-quorum-intersection},
$Q_1 \cap Q_2$ contains at least one correct process.

A correct process only sends messages associated with its current view and it never decreases its current view number.
Hence, it cannot send the vote in view $v$ before sending the acknowledgment in view $w'$.
If the correct process acknowledged value $y$ in $w'$ before sending the vote to the leader of view $v$,
the vote would have contained a view number at least as large as $w'$.
This contradicts the choice of $w$ to be the maximal view contained in a valid vote received by the leader.
\end{proof}

\begin{lemma} \label{lem:no-equivocation-safety}
\atreplace{If there is a unique value $x$ with a valid vote $(x,w,\sigma,\tau)$, then $x$ is safe in view $v$.}{If among the received votes there is only one value $x$ such that there is a valid vote for $x$ in view $w$ $(x, w, \sigma, \tau)$, then $x$ is safe in view $v$.}
\end{lemma}
\begin{proof}
Suppose, by contradiction, that at some point of the execution some other value $y$ is decided in a view $w'$ smaller than $v$.
By the induction hypothesis, $x$ is safe in $w$, and thus $w'$ cannot be smaller than $w$.
By Lemma~\ref{lem:no-decisions-between-views}, $w'$ cannot be larger than $w$.
Let us consider the remaining case $(w' = w)$.


Let $Q_1 \subset \AllProc$ be the set of $n-f$ processes that acknowledged value $y$ in $w$.
Let $Q_2 \subset \AllProc$ be the set of $n-f$ processes that sent their votes to the leader of view $v$.
By~\ref{prop:simple-quorum-intersection}, $Q_1 \cap Q_2$ contains at least one correct process.

A correct process only sends messages associated with its current view and it never decreases its current view number.
Hence, it cannot send the vote in view $v$ before sending the acknowledgment in view $w$.
If the correct process acknowledged value $y$ in $w$ before sending the vote to the leader of view $v$,
the vote would have contained either a view number larger than $w$ (which contradicts the maximality of $w$) 
or the value $y$ (which contradicts the uniqueness of $x$).
\end{proof}

\begin{lemma} \label{lem:equivocation-safety-part-1}
If the leader detects an equivocating process $q$
and receives at least $2f$ valid votes for a value $x$ in view $w$ from processes other than $q$, 
then $x$ is safe in view $v$.
\end{lemma}
\begin{proof}
Suppose, by contradiction, that at some point of the execution some other value ($y$) is decided in a view $w'$ smaller than $v$.
By the induction hypothesis, $x$ is safe in $w$, \atadd{and thus} $w'$ cannot be smaller than $w$.
By Lemma~\ref{lem:no-decisions-between-views}, $w'$ cannot be larger than $w$.
Let us consider the remaining case $(w' = w)$.

Let $Q_1 \subset \AllProc$ be the set of $n-f$ processes that acknowledged value $y$ in $w$.
Let \atreplace{$Q_2 \subset \AllProc$}{$Q_2 \subset \AllProc \setminus \{q\}$} be the set of $2f$ processes that cast votes for value $x$ in view $w$.
Since $q \notin Q_2$ and $q$ is provably Byzantine,
there are at most $f-1$ Byzantine processes in $Q_2$.
By~\ref{prop:quorum-intersection:n-f:2f}, there is at least one correct process in $Q_1 \cap Q_2$.
A correct process only adopts a vote before acknowledging the value from the vote and it never
acknowledges 2 different values in the same view.
Hence, $y = x$---a contradiction.
\end{proof}

\begin{lemma} \label{lem:equivocation-safety-part-2}
If the leader detects an equivocating process $q$
and does not receive $2f$ or more valid votes for any value $x$ in view $w$ from processes other than $q$, 
then any value is safe in $v$.
\end{lemma}
\begin{proof}
Suppose, by contradiction, that at some point of the execution some value $y$ is decided in a view $w'$ smaller than $v$.
Let $m_1=(y_1,w,\sigma_1,\tau_1)$ and $m_2=(y_2,w,\sigma_2,\tau_2)$ be the two valid votes such that $y_1 \neq y_2$.
By the induction hypothesis \atadd{and the validity of certificate $\sigma_1$}, no value other than $y_1$ was or will ever be decided in a view smaller than $w$.
The same applies for value $y_2$.
\atreplace{Hence}{Since $y_1 \neq y_2$}, no value was or will ever be decided in a view smaller than $w$ (i.e., $w'$ is not smaller than $w$).
By Lemma~\ref{lem:no-decisions-between-views}, $w'$ is not larger than $w$.

Let us consider the remaining case $(w' = w)$.
Recall that the leader collects $n-f$ votes from processes other than $q$.
By~\ref{prop:quorum-intersection:n-f:n-f}, the leader would have received at least $2f$ votes \atadd{from processes other than $q$} for the value $y$ in view $w$ or at least one vote for a value in a view larger than $w$.
\end{proof}

\begin{theorem} \label{the:consistency}
The proposed algorithm satisfies the consistency property of consensus.
\end{theorem}
\begin{proof}
Suppose, by contradiction, that two processes decided on different values $x$ and $y$, in views $v$ and $v'$, respectively.
Without loss of generality, assume that $v\geq v'$.
By Lemmata~\ref{lem:safe-if-all-nil}-\ref{lem:equivocation-safety-part-2} and the algorithm for choosing a value to propose, value $x$ can only be decided in view $v$ if it is safe in $v$. 
As no value other than $x$ can be decided in a view less than $v$, we have $v=v'$. 
But a value can only be decided in a view if $n-f$ processes sent {\MAck} messages in that view. 
By~\ref{prop:simple-quorum-intersection}, at least one correct process must have sent {\MAck} messages for both $x$ and $y$ in view $v$---a contradiction with the algorithm.
\end{proof}

\subsection{Generalized version}


\atrev{%
Following the example of previous work~\cite{kursawe2002optimistic,fab-paxos,revisiting-fast-bft-2}, we can add a PBFT-like \emph{slow path} in order to obtain a generalized version of our algorithm.
The protocol will tolerate $f$ Byzantine failures and will be able to decide a value in the common case after just two communication steps as long as the actual number of faults does not exceed threshold $t$ ($1 \le t \le f$).
The required number of processes will be $3f+2t-1$.
Note that, when $t=1$, we obtain a Byzantine consensus protocol with optimal resilience ($3f+2t-1=3f+1$ when $t=1$) that is able to decide a value with optimal latency in the common case in presence of a single Byzantine fault.
To the best of our knowledge, in all previous algorithms with optimal resilience ($n=3f+1$), the optimistic fast path could make progress only when all processes were correct.

We describe the generalized version of our protocol in detail in Appendix~\ref{app:generalized-version}.}

\section{Lower Bound} \label{sec:lower-bound}




In this section, we show that any $f$-resilient Byzantine consensus protocol that terminates within two message delays in the common case when the number of actual failures does not exceed $t$ (we call such a protocol \emph{$t$-two-step}) requires at least $3f+2t-1$ processes.

%
%

In Section~\ref{subsec:fab-lower-bound}, we also show that the higher lower bound of $n=3f+2t+1$ processes (claimed by Martin and Alvisi~\citep{fab-paxos}) holds for a special class of protocols assuming that the processes that propose values (so called \emph{proposers}) are disjoint from the processes responsible for replicating the proposed values (so called \emph{acceptors}).

\subsection{Preliminaries}

It is well-known that a partially-synchronous Byzantine consensus requires at least $3f+1$ processes~\cite{pease1980reaching}.
Hence, we assume that $|\AllProc| \ge 3f+1$.
\atrev{Additionally, since the case when $f=0$ is trivial, in the rest of this section, we assume that $f \ge 1$.}
%
%

Let $\Domain$ be the domain of the consensus protocol (i.e., the set of possible input values).
We define an \emph{initial configuration} as a function 
$\IConf: \AllProc \to \Domain$ that maps processes to their input values.
Note that although $\IConf$ maps all processes to some input values, Byzantine processes can pretend as if they have different input\atadd{s}.

An \emph{execution} of the protocol is the tuple $(\IConf, \ByzProc, \Schedule)$,
where $\IConf$ is an initial configuration,
$\ByzProc$ is the set of Byzantine processes ($|\ByzProc| \le f$),
and $\Schedule$ is a totally ordered sequence of steps taken by every process consisting of ``send message'', ``receive message'', and ``timer elapsed'' events.
Each event is marked with an absolute time when it occurred, according to the global clock.
We allow multiple events to happen at the same time, but $\Schedule$, nevertheless, arranges them in a total order.
If $\exec = (I, \ByzProc, \Schedule)$, we say that execution $\exec$ \emph{starts from} initial configuration $\IConf$.

%
\atrev{%
In the proof of this lower bound, 
we assume that all processes have access to perfectly synchronized local clocks that show exact time elapsed since the beginning of the execution.
Note that this only strengthens our lower bound.
If there is no algorithm implementing fast Byzantine consensus with $3f+2t-2$ or fewer processes in the model with perfectly synchronized clocks, then clearly there is no such algorithm with loosely synchronized clocks.}

%
We refer to events that happen during the half-open time interval $[0, \Delta)$ as \emph{the first round},
to the events that happen during the half-open time interval $[\Delta, 2\Delta)$ as \emph{the second round}, and so on.
%
In all executions that we consider, a message sent in round $i$ will be delivered in round $i+1$ or later.
%



\begin{lemma} \label{lem:first-round-determinism}
Actions taken by correct processes during the first round depend exclusively on their inputs (i.e., on the initial configuration).
\end{lemma}
\begin{proof}
Indeed, in the executions that we consider, during the first round, no messages can be delivered. 
Messages that are sent at time $0$ are delivered not earlier than at time $\Delta$, which belongs to the second round.
As we only consider deterministic algorithms, 
all actions taken by the processes in the first round are based on their input values.
\end{proof}


Thanks to the {\liveness} property of consensus, we can choose to only consider \atremove{finite} executions in which every correct process decides on some value at some point.
Moreover, by the {\consistency} property of consensus, all correct processes have to decide the same value.
Let us call this value the \emph{consensus value} of an execution and denote it with $c(\exec)$, where $\exec$ is an execution.

Given an execution $\exec$ and a process $p$, the
\emph{\decisionView} of $p$ in $\exec$
is the view of $p$ at the moment when it triggers the $\fDecide$ callback.
The view consists of the messages $p$ received (ordered and with the precise time of delivery)
together with the state of $p$ in the initial configuration of $\exec$.
Note that the messages received by $p$ \emph{after} it triggers the callback are not reflected in the {\decisionView}. 

Let $\exec_1$ and $\exec_2$ be two executions, and let $p$
be a process which is correct in $\exec_1$ and $\exec_2$. 
Execution $\exec_1$ is \emph{similar} to execution $\exec_2$ with respect to $p$, 
denoted as $\exec_1 \simFor{p} \exec_2$, 
if the {\decisionView} of $p$ in $\exec_1$ is the same as the {\decisionView} of $p$ in $\exec_2$.
If $P$ is a set of processes, we use $\exec_1 \simFor{P} \exec_2$ as a shorthand for $\forall p \in P: \exec_1 \simFor{p} \exec_2$.

\begin{lemma} \label{lem:similar-executions}
If there is a correct process $p \in \AllProc$ such that $\exec_1 \simFor{p} \exec_2$, then $c(\exec_1) = c(\exec_2)$.
\end{lemma}
\begin{proof}
Since we only consider executions where all correct processes decide some value, in executions $\exec_1$ and $\exec_2$, process $p$ had to decide values $c(\exec_1)$ and $c(\exec_2)$ respectively.
However, since, at the moment of the decision, process $p$ is in the same state in both executions and we only consider deterministic processes, $p$ has to make identical decisions in the two executions. Hence, $c(\exec_1) = c(\exec_2)$.
\end{proof}

We say that $\exec = (\IConf, \ByzProc, \Schedule)$ is a \emph{$\T$-faulty two-step execution},
where $\T \subset \AllProc$ and $|\T| = t$, iff:
\begin{enumerate}
    \item All processes in $\AllProc \setminus \T$ are correct and all processes in $\T$ are Byzantine
    (i.e., $\ByzProc = \T$);
    
    \item 
    All messages sent during round $i$ (i.e., time interval $[(i-1)\Delta, i\Delta)$) are delivered precisely at the beginning of the next round (i.e., at time $i\Delta$).
    
    \item Local computation is instantaneous.
    In particular, if a correct process receives a message at time $T \in [(i-1)\Delta, i\Delta)$ and sends a reply without explicit delay, the reply will be sent also at time $T$ and will arrive at time $i\Delta$;
    
    \item The Byzantine processes in $\T$ correctly follow the protocol with respect to the initial configuration $\IConf$ during the first round.
    After that, they stop taking any steps.
    In particular, they do not send any messages at time $\Delta$ or later;
    
    \item Every correct process makes a decision not later than at time $2\Delta$.
\end{enumerate}

\atadd{Intuitively, a $\T$-faulty two-step execution is an execution with relatively favorable conditions (the system is synchronous from the beginning and the Byzantine processes fail by simply crashing at time $\Delta$) in which all correct processes decide after just two message delays.}

A protocol $\Protocol$ is called \emph{$t$-two-step} if it satisfies the following conditions:
\begin{enumerate}
    \item $\Protocol$ is a consensus protocol with {\weakValidity}, as defined in Section~\ref{sec:preliminaries};
    
    \item
    %
    For all $\T \subset \AllProc$ of size $t$,
    there is a $\T$-faulty two-step execution starting from $\IConf$.%
    \footnote{In Section~\ref{subsec:weaken}, we discuss ways to relax this assumption.}
\end{enumerate}
%
%
In other words, if there are at most $t$ Byzantine processes that fail simply by crashing at time $\Delta$, local computation is immediate, and the network is synchronous, it must be possible for all processes to decide after just $2$ steps.
Otherwise, when the environment is not so gracious (e.g., the network is not synchronous from the beginning or there are more than $t$ Byzantine processes), the protocol is allowed to terminate after more than $2$ steps.

The protocol presented in Section~\ref{sec:algorithm} is $t$-two-step.
Indeed, suppose that we have at least $3f + 2t - 1$ processes and $f \ge 1$.
Recall that $\fleader(1)$ is the leader for view $1$. 
Let $p = \fleader(1)$ \atadd{and let $\T$ be an arbitrary set of $t$ processes}.
Then, for any initial configuration and any set $\T$ of $t$ processes, the following $\T$-faulty two-step execution exists: 
\begin{enumerate}
    \item $p$ proposes its input value $x = \xinput_p$ at time $0$ with the message $\MPropose(x,1,\bot,\tauHat)$; 

    \item All other processes, including those in $\T \setminus \{p\}$, do nothing during the first round;

    \item At time $\Delta$, every process receives the propose message, and \atreplace{$3f+t-1$}{$n-t$} correct processes respond with an acknowledgment $\MAck(x, 1)$;

    \item At time $2\Delta$, every correct processes receives \atreplace{all the}{the $n-t$} $\MAck$ messages and decides.
\end{enumerate}

The following lemma explains how the {\weakValidity} property of consensus dictates the output values of $\T$-faulty two-step executions.
%
\begin{lemma} \label{lem:f-faulty-validity}
    For any consensus protocol with {\weakValidity},
    if all processes have the same input value $x$ $(\forall p: \IConf(p) = x)$,
    for any $\T$-faulty two-step execution $\exec$ starting from $\IConf$,
    the consensus value $c(\exec)$ equals $x$. 
\end{lemma}
\begin{proof}
    Consider a $\T$-two-step execution $\exec$, and let $T$ be the moment of time by which every correct process has invoked the $\fDecide$ callback.
    Let $\exec'$ be an execution identical to $\exec$, except
    that in $\exec'$, processes in $\T$ are not Byzantine, but just slow.
    The messages they send at time $\Delta$ or later reach the other processes only after time $T$.
    Notice that the processes in $\AllProc \setminus \T$ have no way to distinguish $\exec'$ and $\exec$ until they receive the delayed messages,
    which happens already after they decide.
    Hence, $\exec' \simFor{\AllProc \setminus \T} \exec$ and, by Lemma~\ref{lem:similar-executions}, $c(\exec') = c(\exec)$.
    By the {\weakValidity} property of consensus, if all processes have $x$ as their input value in $\exec'$, then $c(\exec') = x$.
\end{proof}

\subsection{Optimality of our algorithm}

Process $p \in \AllProc$ is said to be \emph{\influential} if there are two initial configurations ($\IConf$ and $\IConf'$)
such that $\forall q \neq p: \IConf(q) = \IConf'(q)$
and two non-intersecting sets of processes not including $p$ of size $t$ ($\T, \T'\subset \AllProc \setminus \{p\}$, $|\T| = |\T'| = t$, and $\T \cap \T' = \emptyset$)
such that there is a $\T$-faulty execution $\exec$ 
and a $\T'$-faulty execution $\exec'$ with different consensus values ($c(\exec) \neq c(\exec')$).
%
    
    
    

Intuitively, a process is {\influential} if its input value under certain circumstances can affect the outcome of the fast path of the protocol.
%
In Theorem~\ref{the:lower-bound}, we prove that, if the number of processes is smaller than $3f+2t-1$, an {\influential} process can use its power to force disagreement.
%

\begin{lemma} \label{lem:influential-process}
For any $t$-two-step consensus protocol, there is at least one {\influential} process.
\end{lemma}
\begin{proof}
$\forall i \in \{0, \ldots, n\}$: let $\IConf_i$ be the initial configuration in which the first $i$ processes have the input value $1$ and the remaining processes have the input value $0$.
In particular, in $\IConf_0$, all processes have the input value $0$, and, in $\IConf_n$, all processes have the input value $1$.
By the definition of a $t$-two-step consensus protocol,
for all $i \in \{1, \dots, n\}$ and $\T \subset \AllProc$ ($|\T| = t$),
there must be a $\T$-faulty two-step execution starting from $\IConf_i$.
Moreover, by Lemma~\ref{lem:f-faulty-validity},
all $\T$-faulty two-step executions starting from $\IConf_0$ (resp., $\IConf_n$) have the consensus value $0$ (resp., $1$).
\atreplace{Let $\AllProc$ be $\{p_1, \dots, p_n\}$.}{Recall that $\AllProc = \{p_1, \dots, p_n\}$.}
For all $i \in \{1, \dots, n\}$, let $\pred(i)$ be the predicate ``there is a set $\T_1 \subset (\AllProc \setminus \{p_i\})$ such that there is a $\T_1$-faulty two-step execution with consensus value $1$ starting from $\IConf_i$''.
Let $j$ be the minimum number such that $\pred(j) = \True$ (note that such a number exists because $\pred(n) = \True$).
Let $\T_1$ be the set of processes defined in the predicate. By definition, $p_j \notin \T_1$.

Let us consider two cases:
\begin{itemize}
    \item If $j > 1$, let $\T_0$ be an arbitrary subset of $\AllProc \setminus (\T_1 \cup \{p_j, p_{j-1}\})$ of size $t$.
    Note that such a subset exists because $|\AllProc \setminus (\T_1 \cup \{p_{j-1}, p_j\})| = |\AllProc| - (t + 2) \ge (3f+1) - (t + 2) \ge t$.
    Since $j$ is the minimum number such that $\pred(j) = \True$, $\pred(j-1) = \False$.
    Hence, all $\T_0$-faulty two-step executions starting from initial configuration $\IConf_{j-1}$ have consensus value $0$.
    By the definition of a $t$-two-step consensus protocol, there is at least one such execution.
    \atrev{Let $\exec_0$ be such an execution.}
    
    \item If $j = 1$, let $\T_0$ be an arbitrary subset of $\AllProc \setminus (\T_1 \cup \{p_j\})$ of size $t$.
    By Lemma~\ref{lem:f-faulty-validity}, all $\T_0$-faulty two-step executions starting from $\IConf_0$ have consensus value $0$, and, by the definition of a $t$-two-step consensus algorithm, there is at least one such execution.
    \atrev{Let $\exec_0$ be such an execution.}
\end{itemize}

We argue that $p_j$ is an {\influential} process.
Indeed, $\IConf_{j-1}$ and $\IConf_j$ differ only in the input of process $p_j$, $\rho_0$ and $\rho_1$ are $\T_0$- and $\T_1$-faulty executions starting from $\IConf_{j-1}$ and $\IConf_j$ respectively, $\T_0 \cap \T_1 = \emptyset$, $p_j \notin (\T_0 \cup \T_1)$, and $c(\rho_0) \neq c(\rho_1)$.
\end{proof}


\begin{figure}
    \centering
    \begin{tikzpicture}[xscale=0.75,yscale=0.7]
        \draw[step=1cm,color=black,xshift=0.5cm, yshift=-0.5cm] (-0.5,-1) grid (7,5.5);
        \node at (1, 5){$\{p\}$};
        \node at (2, 5){$P_1$};
        \node at (3, 5){$P_2$};
        \node at (4, 5){$P_3$};
        \node at (5, 5){$P_4$};
        \node at (6, 5){$P_5$};
        \node at (7, 5){$c(-)$};
        
        \node at (0, 0){$\exec_1$};
        \node at (0, 1){$\exec_2$};
        \node at (0, 2){$\exec_3$};
        \node at (0, 3){$\exec_4$};
        \node at (0, 4){$\exec_5$};
        
        \node at (1, 1){$\ninj$};
        \node at (1, 2){$\ninj$};
        \node at (1, 3){$\ninj$};
        
        \node at (2, 0){$\ninj$};
        \node at (2, 1){$s_1$};
        \node at (2, 2){$s_1$};
        \node at (2, 3){$s_1$};
        \node at (2, 4){$s_1$};
        
        \node at (3, 0){$t_2$};
        \node at (3, 1){$\ninj$};
        \node at (3, 2){$s_2$};
        \node at (3, 3){$s_2$};
        \node at (3, 4){$s_2$};
        
        \node at (4, 0){$t_3$};
        \node at (4, 1){$t_3$};
        \node at (4, 2){$\ninj$};
        \node at (4, 3){$s_3$};
        \node at (4, 4){$s_3$};
        
        \node at (5, 0){$t_4$};
        \node at (5, 1){$t_4$};
        \node at (5, 2){$t_4$};
        \node at (5, 3){$\ninj$};
        \node at (5, 4){$s_4$};
        
        \node at (6, 0){$t_5$};
        \node at (6, 1){$t_5$};
        \node at (6, 2){$t_5$};
        \node at (6, 3){$t_5$};
        \node at (6, 4){$\ninj$};
        
        \node at (7, 0){$1$};
        \node at (7, 1){?};
        \node at (7, 2){?};
        \node at (7, 3){?};
        \node at (7, 4){$0$};

        \draw [black] (7.6,0.1) .. controls (7.9,0.3) and (7.9,0.7) .. (7.6,0.9);
        \node [anchor=west] at (7.8, 0.5){$\exec_1 \simFor{P_3} \exec_2$};
        \draw [black] (7.6,1.1) .. controls (7.9,1.3) and (7.9,1.7) .. (7.6,1.9);
        \node [anchor=west] at (7.8, 1.5){$\exec_2 \simFor{P_1 \cup P_4 \cup P_5} \exec_3$};
        \draw [black] (7.6,2.1) .. controls (7.9,2.3) and (7.9,2.7) .. (7.6,2.9);
        \node [anchor=west] at (7.8, 2.5){$\exec_3 \simFor{P_1 \cup P_2 \cup P_5} \exec_4$};
        \draw [black] (7.6,3.1) .. controls (7.9,3.3) and (7.9,3.7) .. (7.6,3.9);
        \node [anchor=west] at (7.8, 3.5){$\exec_4 \simFor{P_3} \exec_5$};
        
        
        \node [anchor=east] at (0.3, -1){size};
        \node at (1, -1){$1$};
        \node at (2, -1){$t$};
        \node at (3, -1){$f{-}1$};
        \node at (4, -1){$f{-}1$};
        \node at (5, -1){$f{-}1$};
        \node at (6, -1){$t$};
    \end{tikzpicture}
    
    \caption{The proof setup of the lower bound when $f \ge t \ge 2$. The rows are executions and the columns are groups (subsets) of the processes. Byzantine groups are denoted with $\ninj$. The states of the processes after the first round are denoted with $s_i$ and $t_i$.}
    \label{fig:lower-bound}
\end{figure}

\begin{figure}[htbp]
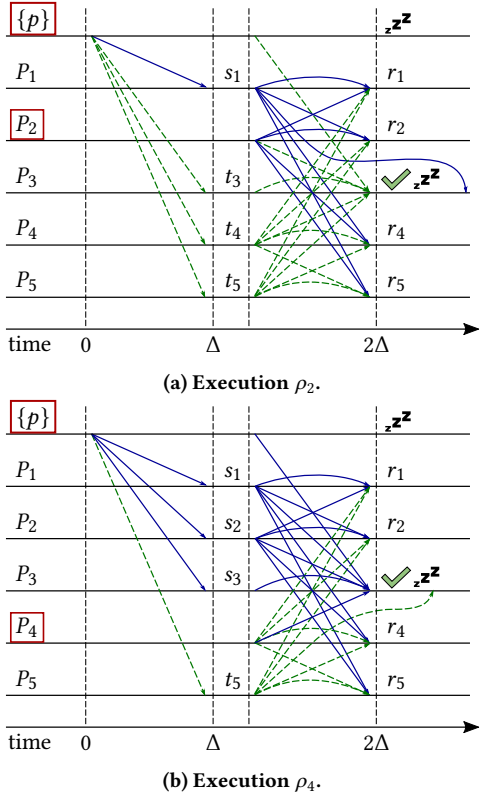

    \centering

    \begin{subfigure}[b]{\linewidth}
        \centering
        \includesvg{svg/lower-bound-exec-2.svg}
        \caption{Execution $\exec_2$.}
        \label{fig:lower-bound-exec-2}
    \end{subfigure}

    \begin{subfigure}[b]{\linewidth}
        \centering
        \includesvg{svg/lower-bound-exec-4.svg}
        \caption{Execution $\exec_4$.}
        \label{fig:lower-bound-exec-4}
    \end{subfigure}


    \caption{
        First two rounds of executions $\exec_2$ and $\exec_4$.
        Solid blue lines and dashed green arrows represent messages identical to messages sent in $\exec_5$ and $\exec_1$ respectively.
        Green tick symbol means that all processes in the group decide a value.
        Messages from all processes other than $p$ in the first round are identical in all $5$ executions and omitted on the picture for clarity.
        Messages sent in the second round to process $p$ are also omitted.
        $s_i$ and $t_i$ represent states of correct processes after the first round in $\exec_5$ and $\exec_1$ respectively.
        States of correct processes after the second round in $\exec_3$ are denoted by $r_i$.
        In $\exec_2$ (resp, $\exec_4$), processes in group $P_2$ (resp., $P_4$) are Byzantine, but they pretend to be correct and in state $r_2$ (resp., $r_4$).}
     
    \label{fig:lower-bound-executions}
\end{figure}

\begin{figure}[htbp]
    \centering

    \includesvg{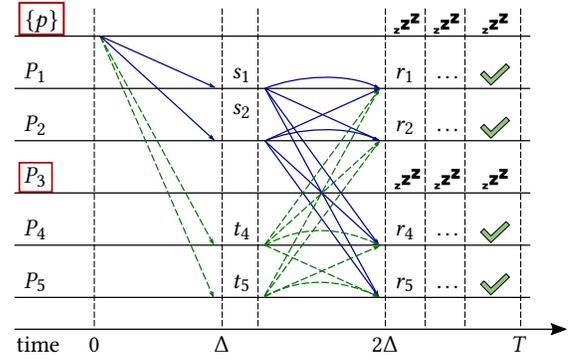}
    
    \caption{%
        Execution $\exec_3$.
        Messages sent in the second round to group $P_3$ are omitted as these processes are Byzantine and do not take any further steps after the second round.}
     
    \label{fig:lower-bound-exec-3}
\end{figure}

\begin{theorem} \label{the:lower-bound}
There is no $f$-resilient $t$-two-step consensus protocol for $3f+2t-2$ processes.
\end{theorem}
\begin{proof}
\atrev{Note that, if $t \le 1$, $3f+2t-2 \le 3f$.
Hence, the case when $t\le1$ follows directly from the classic bound of $3f+1$ processes for partially synchronous $f$-resilient Byzantine consensus~\cite{BT85}.}

Consider the case $t \ge 2$ and suppose, by contradiction, that there is a $t$-two-step consensus protocol
for $3f+2t-2$ processes ($f \ge t \ge 2$).
By Lemma~\ref{lem:influential-process}, there is an {\influential} process $p$, i.e., there exist two initial configurations ($\IConf'$ and $\IConf''$) that differ only in the input of process $p$, two sets of processes ($\T', \T'' \subset \AllProc \setminus \{p\}$, $|\T'| = |\T''| = t$, and $\T' \cap \T'' = \emptyset$), and two executions: a $\T'$-faulty execution $\exec'$ starting from $\IConf'$ and a $\T''$-faulty execution $\exec''$ starting from $\IConf''$, such that $c(\exec') \neq c(\exec'')$.
Without loss of generality, let us assume that $c(\exec') = 0$ and $c(\exec'') = 1$.

We partition the set $\AllProc \setminus \{p\}$ into five groups: $P_1, \ldots, P_5$, where $P_1 = \T''$, $P_5 = \T'$, and $|P_2| = |P_3| = |P_4| = f-1$.
The partition is depicted in Figure~\ref{fig:lower-bound}.
Note that $|P_1| + \ldots + |P_5| + |\{p\}| = 2t + 3(f-1) + 1 = 3f+2t-2 = |\AllProc|$.
%


We construct $5$ executions, $\exec_1$, $\ldots$, $\exec_5$ ($\exec_1=\exec''$ and $\exec_5=\exec'$), such that for all $i \in \{1, \dots, 5\}$, group $P_i$ is Byzantine in $\exec_i$, and for all $j \neq i$, group $P_j$ is correct in $\exec_i$. The influential process $p$ is Byzantine for $\exec_2, \exec_3, \exec_4$ and is correct for $\exec_1$ and $\exec_5$; hence we have exactly $f$ Byzantine processes in each execution as $|P_2| = |P_3| = |P_4| = f-1$. Our goal is to show that each pair of adjacent executions will be similar for at least one correct process \atremove{set}, \atrev{which} will then decide the same value. This would lead to a contradiction, since Lemma~\ref{lem:similar-executions} would imply that $0 = c(\exec_5) = \dots = c(\exec_1) = 1$.

%

%
%
Let $m_i^1$ and $m_i^5$ denote the messages that $p$ sends to $P_i$ in execution $\exec_1$ and $\exec_5$ respectively.
%
%
By Lemma~\ref{lem:first-round-determinism}, all actions taken by correct processes other than $p$ during the first round will be the same in all executions.
For each $i$, during the first round of $\exec_i$, the processes in Byzantine group $P_i$ will act as if they were correct.
Hence, the only process that acts differently in different executions during the first round is $p$.

In each execution $\exec_i$ ($i \in \{2,3,4\}$), $p$ equivocates by sending $m_j^5$ to the processes in $P_j$ for $j < i$ and  $m_j^1$ to the processes in $P_j$ for $j > i$.
Processes in $P_i$ are Byzantine and $p$ sends both types of messages to them so that they can choose what to relay to other processes.
%
In executions $\exec_1$ and $\exec_5$, $p$ honestly follows the protocol.
%
%
%

%
In all executions that we construct, all messages sent during the first round are delivered at time $\Delta$.
If $p$ sends $m_j^5$ (resp., $m_j^1$) to $P_j$ in the first round, then we schedule all message delivery events for $P_j$ in exactly the same order as in $\exec_5$ (resp., $\exec_1$).
Thus, after processing all events at time $\Delta$, each group $P_j$ takes one of two states $s_j$ or $t_j$ (where $s_j$ is the state of $P_j$ after the first round in $\exec_5$ and $t_j$ is the state of $P_j$ after the first round in $\exec_1$) depending entirely on which of the two messages they receive from $p$
(see Figure~\ref{fig:lower-bound}).
Moreover, the Byzantine processes in executions $\exec_2$, $\exec_3$, and $\exec_4$ can choose to pretend as if they are in one of the two states ($s_j$ or $t_j$).
We now describe $\exec_2, \exec_3$, and $\exec_4$ in more detail.

\paragraph{Execution $\exec_4$, second round.}
%
The first two rounds of execution $\exec_4$ are depicted in Figure~\ref{fig:lower-bound-exec-4}.
The second round is identical to $\exec_5$ (including the time at which events happen and their order) with the following modifications:
\begin{itemize}
    \item $p$ is now Byzantine. It sends messages to $P_3$ in the same fashion as in $\exec_5$ and is silent to other processes;
    
    \item $P_4$ is now Byzantine. They send messages to $P_3$ in the same fashion as in $\exec_5$ (i.e., as if they are correct and were in the state $s_4$ after the first round). For other processes, $P_4$ acts in exactly the same fashion as in $\exec_1$.
    
    \item $P_3$ is now slow (but still correct). It sends the same messages as in $\exec_5$ but they are not to be received by any other process until a finite time $T$ that we will specify later;
    
    %
    \item $P_5$ is now correct. The messages from $P_5$ to $P_3$ are delayed and do not reach the recipients until after time $2\Delta$; other messages from $P_5$ are delivered in a timely fashion.
    
\end{itemize}
We now look at $P_3$'s perspective. During time interval $[0, 2\Delta]$,
$P_3$ will not be able to distinguish this execution from $\exec_5$, since it receives the exact same messages from $\{p\}, P_1, P_2, P_4$, which all have the same state as in $\exec_5$ after the first round (or in $P_4$'s case can fake the same state), and hears nothing in the second round from $P_5$ in both executions.
Thus, by time $2\Delta$, processes in $P_3$ will achieve exactly the same state as in $\exec_5$ and will decide $0$ as well (note that this decision is done in silence, as $P_3$'s messages will not be received by anyone else until  time $T$). Therefore, $\exec_5 \simFor{P_3} \exec_4$. 


\paragraph{Execution $\exec_2$, second round.}
The execution depicted in Figure~\ref{fig:lower-bound-exec-2} is similar to $\exec_4$, except that now the set of Byzantine processes is $\{p\} \cup P_2$, and they send messages to group $P_3$ in exactly the same fashion as in $\exec_1$. 
Furthermore, now round-$2$ messages from $P_1$ reach $P_3$ at time $T>2\Delta$. 
%
Using an argument symmetric to the one used for $\exec_4$ above, we conclude that by time $2\Delta$, processes in $P_3$ will achieve exactly the same state as in $\exec_1$ and will decide $1$ as well.
Therefore, $\exec_1 \simFor{P_3} \exec_2$.

\paragraph{Execution $\exec_3$.}
%
Note that, in execution $\exec_2$, exactly the same set of messages is sent in the second round to non-$P_3$ processes, as in execution $\exec_4$.
\atrev{We can set the time and the order of their delivery to be identical as well.}
%
%
In the subsequent rounds, unless some message from $P_3$ is delivered to a non-$P_3$ correct process, Byzantine processes in $P_2$ in $\exec_2$ can act in exactly the same way as their correct counterparts in $\exec_4$.
Similarly, Byzantine processes in $P_4$ in $\exec_4$ can act in exactly the same way as their correct counterparts in $\exec_2$.
Hence, the two executions will remain indistinguishable for $P_1$ and $P_5$ until some message from $P_3$ is delivered.
%
%

To show that the non-$P_3$ correct processes must decide some value without waiting for the messages from $P_3$, we construct execution $\exec_3$ \atadd{(depicted in Figure~\ref{fig:lower-bound-exec-3})}.
In $\exec_3$, $p$ sends $m_j^5$ to the processes in $P_j$ for $j<3$ and $m_j^1$ to the processes in $P_j$ for $j>3$.
This results in all processes except those in $\{p\} \cup P_3$ acting in the same way during the second round as in $\exec_2$ and $\exec_4$.
Processes in $P_3$ are Byzantine and fail by crashing at time $\Delta$, before sending any messages in the second round.
By the {\liveness} property of consensus, there must exist a continuation in which every correct process decides a value.
Let $\rho_3$ be an arbitrary such continuation and let $T$ be an arbitrary moment in time after all correct processes made their decisions.

\paragraph{Executions $\exec_2$ and $\exec_4$, later rounds.}
We are now finally ready to complete executions $\exec_2$ and $\exec_4$. Since they are symmetric, we start by looking at $\exec_4$. 
We have already described $\exec_4$ up to time $2\Delta$.
Recall that messages from $P_3$ are delayed until time $T$ that we specified earlier.
As the resulting execution is identical to $\exec_3$ until time $T$, all correct processes (in particular, all processes in $P_1$) must decide. Thus, $\exec_4 \simFor{P_1} \exec_3$.
By a symmetric argument, $\exec_2 \simFor{P_5} \exec_3$. 
%
As $\exec_1 \simFor{P_3} \exec_2$ and $\exec_4 \simFor{P_3} \exec_5$, we have established that every two adjacent executions are similar---a contradiction.
%
\end{proof}


\subsection{Weakening the assumptions}
\label{subsec:weaken}

In the definition of a $t$-two-step consensus protocol, we require that for all $\T \subset \AllProc$ of size $t$, there exists a $\T$-faulty two-step execution of the protocol.
This may appear counter-intuitive as all existing fast Byzantine consensus algorithms (including the one proposed in this paper) are leader-based: they guarantee fast termination only in case when the leader is correct.
But there is no contradiction here, as in a $\T$-faulty two-step execution, the processes in $\T$ honestly follow the protocol during the first round and crash only at time $\Delta$.
In all fast Byzantine consensus protocols that we are aware of, the correctness of the leader during the first round is sufficient to reach consensus in two steps.

Nevertheless, we cannot exclude the possibility that some protocols may rely on the leader's participation in the second round as well.
To encompass such protocols, we can require $\T$ to be selected from $\AllProc \setminus \{p\}$, where $p$ is a designated process (the leader of the first view).
More generally, we can introduce a set of ``suspects'' $\Suspects \subset \AllProc$ and only require $\T$-faulty two-step executions to exist for all $\T \subset \Suspects$ of size $t$.
This set $\Suspects$ should be of size at least $2t+2$ (we will see why this is important shortly).
Hence, the fast path of the protocol can rely on $n-(2t+2)$ ``leaders''.
Note that, since $t \le f$, when $f \ge 2$, $n \ge 3f+1 \ge (f + 2t) + 1 \ge 2t + 3$.%
\footnote{Recall that when $f \le 1$, the lower bound trivially follows from the fact that any partially-synchronous Byzantine consensus algorithm needs at least $3f+1$ processes.}
Hence, there is always at least one ``non-suspect''.

The only argument that we will have to modify is in the proof of Lemma~\ref{lem:influential-process}.
There, we will have to select $\T_0$ and $\T_1$ out of $\Suspects$, and not just $\AllProc$.
To be able to select $\T_0$, we require $|\Suspects \setminus (\{p_j, p_{j-1}\} \cup \T_1)| \ge t$.
This holds if and only if $|\Suspects| \ge 2t+2$.

\subsection{Optimality of FaB Paxos} \label{subsec:fab-lower-bound}




While $n=3f+2t+1$ is not optimal for fast Byzantine consensus algorithms in general,
it is optimal for a special class of Paxos-like algorithms that separate proposers from acceptors.
In Paxos~\cite{paxos}, one of the first crash fault-tolerant solutions for the consensus problem, Leslie Lamport suggested a model with three distinct types of processes: \emph{proposers}, \emph{acceptors}, and \emph{learners}.
Proposers are ``leaders'' and they are responsible for choosing a safe value and sending it to acceptors.
Acceptors store the proposed values and help the new leader to choose a safe value in case previous leader crashes.
Finally, learners are the processes that trigger the $\fDecide$ callback and use the decided value (e.g., they can execute replicated state machine commands).
In this model, the consensus problem requires all learners to decide the same value.
The Byzantine version of Paxos~\cite{byzantine-paxos} requires presence of at least one correct proposer and $n = 3f + 1$ acceptors, where $f$ is the possible number of Byzantine faults among acceptors.
%

In our algorithm, when a correct leader (proposer) sees that some previous leader equivocated, it uses this fact to exclude one acceptor from consideration as it is provably Byzantine.
This trick only works when the set of proposers is a subset of the set of acceptors.
Moreover, this trick seems to be crucial for achieving the optimal resilience ($n=\max\{3f+2t-1, 3f+1\}$).
When the set of proposers is disjoint from the set of acceptors, or even if there is
just one proposer that is not an acceptor, it can be shown that $n=3f+2t+1$ is optimal.

In order to obtain the $n=3f+2t+1$ lower bound for the model where proposers are separated from acceptors,
we need to make just two minor modifications to our proof of theorem~\ref{the:lower-bound}.
First of all, the {\influential} process $p$ is no longer an acceptor.
Hence, we are left with only 5 groups of acceptors ($P_1, \dots, P_5$) instead of 6 ($\{p\}, P_1, \dots, P_5$).
Second, the groups of acceptors $P_2$, $P_3$, and $P_4$ can now be of size $f$ instead of $f-1$ (since $p$ is no longer counted towards the quota of $f$ Byzantine acceptors).
After these two modifications, the proof shows that there is no $t$-two-step consensus protocol with $n = |P_1| + \dots + |P_5| = 3f+2t$ or fewer acceptors. 




\section{Related Work} \label{sec:related-work}

To the best of our knowledge, Kursawe~\cite{kursawe2002optimistic} was the first to implement a fast (two-step) Byzantine consensus protocol. 
The protocol is able to run with $n=3f+1$ processes, but it is able to commit in two steps only when all $n$ processes follow the protocol and the network is synchronous.
Otherwise, it falls back to a randomized asynchronous consensus protocol.

Martin and Alvisi~\cite{fab-paxos} present FaB Paxos -- a fast Byzantine consensus protocol with $n=5f+1$.
Moreover, they present a parameterized version of the protocol: it runs on $n=3f+2t+1$ processes ($t \le f$), tolerates $f$ Byzantine failures, and is able to commit after just two steps in the common case when the leader is correct, the network is synchronous, and at most $t$ processes are Byzantine.
In the same paper, the authors claim that $n=3f+2t+1$ is the optimal resilience for a fast Byzantine consensus protocol.
In this paper, we show that this lower bound only applies to the class of protocols that separate processes that execute the protocol (acceptors) from the process that propose values (proposers).

Bosco~\cite{bosco} is a Byzantine consensus algorithm that is able to commit values after just one communication step when there is no contention (i.e., when all processes propose the same value).
In order to tolerate $f$ failures, the algorithm needs $5f+1$ or $7f+1$ processes, depending on the desired validity property.

Zyzzyva~\cite{zyzzyva}, UpRight~\cite{upright}, and SBFT~\cite{sbft} are practical systems that build upon the ideas from FaB Paxos to provide optimistic fast path.
Zyzzyva~\cite{zyzzyva} and UpRight~\cite{upright} aim to replace crash fault-tolerant solutions in datacenters.
The evaluations in these papers demonstrate that practical systems based on fast Byzantine consensus protocols can achieve performance comparable with crash fault-tolerant solutions while providing additional robustness of Byzantine fault-tolerance.
In~\cite{sbft}, Gueta et al.\ explore the applications of fast Byzantine consensus to permissioned blockchain.
Due to the high number of processes usually involved in such protocols, the results of this paper are less relevant for this setting.

In~\cite{revisiting-fast-bft-1} and~\cite{revisiting-fast-bft-2}, Abraham et al.\ demonstrate and fix some mistakes in FaB Paxos~\cite{fab-paxos} and Zyzzyva~\cite{zyzzyva}.
Moreover, they combine the ideas from the two algorithm into a new one, called Zelma.
This algorithm lies at the core of the SBFT protocol~\cite{sbft}.

In~\cite{hbft}, the authors claim that their protocol, called hBFT, achieves the two-step latency despite $f$ Byzantine failures with only $3f+1$ processes (as opposed to $5f-1$ required by the lower bound in this paper).
However, in a later paper~\cite{revisiting-hbft}, it was shown that hBFT fails to provide the {\consistency} property of consensus.

In a concurrent work~\cite{abraham2021good-podc}, Abraham et al.\ consider the problem of \emph{Byzantine broadcast} in which a designated leader is expected to reliably disseminate its message to a set of processes.
They show that the problem has a partially synchronous $f$-resilient solution with \emph{good-case latency} of $2$ message delays if and only if $n \ge 5f-1$.
This is very similar to our results \atremove{in Section~\ref{sec:lower-bound}} and their arguments are based on the same ideas.
For this specific result, we believe, however, that our lower bound \atadd{(presented in Section~\ref{sec:lower-bound})} is more general as it is not limited to leader-based algorithms
and encompasses double-threshold algorithms that distinguish $f$, the resilience of the algorithm, and $t$, the actual number of faulty processes in an execution. 
On the other hand, in \cite{abraham2021good-podc} (with their complementary note~\cite{abraham2021fast}), Abraham et al. give a complete categorization of \atadd{the} good-case latency of broadcast protocols 
in both partially–synchronous and synchronous models.


\section*{Acknowledgments} 
The authors are grateful to Jean-Philippe Martin and Lorenzo Alvisi for helpful discussions on their work~\cite{fab-paxos}.

\bibliographystyle{ACM-Reference-Format}
\bibliography{main}
 
\clearpage
\appendix
\section{Generalized Version} \label{app:generalized-version}

The generalized version of our protocol is parametrized by two numbers, $f$ and $t$ ($1 \le t \le f$), and requires $n \ge 3f + 2t - 1$ processes.
When $t = f$, it boils down to the vanilla version of our protocol with $n \ge 5f - 1$.
The protocol solves consensus (provides both safety and liveness) despite up to $f$ processes being Byzantine faulty.
Moreover, it guarantees two-step termination in the common case as long as the actual number of failures does not exceed $t$.

\begin{figure}

    \centering
    \includesvg{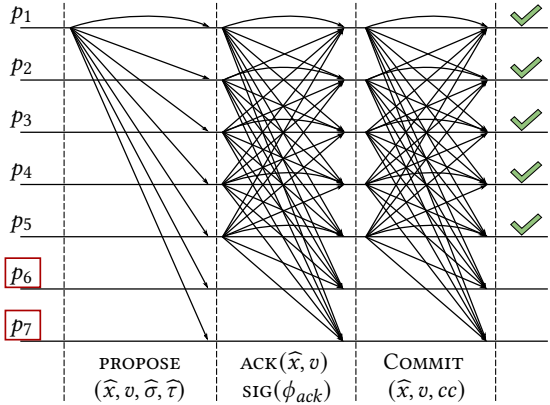}
    
    \caption{Example of a value $\xHat$ committed in view $v$ with the slow path of the generalized protocol with $n=7$, $f=2$, and $t=1$.
    $\sigAck = \fsign_{q}((\MAck, \xHat, v))$, where $q$ is the identifier of the process that sends the $\MAck$ message.}
        
    \label{fig:generalized-normal-case-example}
\end{figure}

\subsection{Proposing a value}

The protocol will have two ways through which a value can be decided by a correct process.
First, there is the \emph{fast path}:
upon receiving $n-t$ $\MAck$ messages for the same value in the same view, the process decides this value.
This is similar to the way values are decided in the non-generalized ($n \ge 5f-1$) version of the protocol.
If the actual number of failures does not exceed $t$, the leader is correct, and synchrony holds, then all correct processes must be able to decide a value through the fast path in just two message delays.
Additionally, there \atreplace{will be}{is} a \emph{slow path} that allows all correct processes to decide a value after three message delays in the common case \atadd{when the actual number of failures is greater than $t$}.

In order to construct the slow path, we introduce two additional message types.
First, every time a correct process sends message $\MAck(x, v)$, it also sends along message $\MAckSig(\sigAck)$, where $\sigAck = \fsign_q((\MAck, x, v))$.
$\lceil \frac{n + f + 1}{2} \rceil$ such signatures for the same value and view constitute a \emph{commit certificate}.
As generating the signature is an expensive operation, we send $\sigAck$ in a separate message in order to avoid slowing down the fast path.
\atreplace{Recall that any}{Any} two sets of processes of size $\lceil \frac{n + f + 1}{2} \rceil$ intersect in at least one correct process.
Hence, it is impossible to collect two commit certificates for different values in the same view.
Also, note that if there is a commit certificate for some value $x$ in view $v$, then no value other than $x$ can be committed through the fast path in view $v$ because any set of size $\lceil \frac{n + f + 1}{2} \rceil$ intersects with any set of size $n-t$ in at least one correct process.

Upon collecting a commit certificate for value $x$ in view $v$, a correct process sends message $\MCommit(x, v, \comCert)$ to everyone, where $\comCert$ is the set of signatures that constitute the commit certificate.
Finally, upon collecting $\lceil \frac{n + f + 1}{2} \rceil$ valid $\MCommit$ messages for value $x$ in view $v$, a correct process decides $x$.

Figure~\ref{fig:generalized-normal-case-example} illustrates an example execution of the slow path.

\subsection{View change}

Finally, we need to modify the view change protocol in order to support the slow path.
As before, the new leader starts the view change by collecting $n-f$ votes, and every vote contains the latest propose message that the process has acknowledged.
However, additionally, each process will add to their vote the latest commit certificate that they have collected.

The rest of the view change protocol is mostly unchanged.
However, in case of a detected equivocation, we have three cases to consider instead of two:
\begin{enumerate}
    \item If there is a commit certificate for value $x$ in view $w$ in $\votes'$, then $x$ is selected; \footnote{%
    Recall that $w$ is the highest view number contained in a valid vote and $\votes'$ is the set of $n-f$ valid votes from processes other than $\fleader(w)$.
    \atadd{The new leader has a proof that $\fleader(w)$ is Byzantine.}}

    \item Otherwise, if there is a set $V \subset \votes'$ of $f+t$ valid votes for a value $x$ in view $w$ in $\votes'$, then $x$ is selected;

    \item Otherwise, any value is safe in view $v$, and the leader simply selects its own input value.
\end{enumerate}

\subsection{Correctness proof}

As with the original $5f-1$ protocol, it is easy to see why the protocol satisfies {\liveness}: once a correct process is elected as a leader after the GST, there is nothing that can stop it from driving the protocol to completion.
The {\extendedValidity} property is immediate as well.
Hence, we focus on {\consistency}.

First, let us prove two simple lemmas about the commit certificate:

\begin{lemma} \label{lem:commit-implies-progress}
    If there is a valid commit certificate for value $x$ in view $u$, 
    then there is a valid progress certificate for value $x$ in view $u$.
\end{lemma}
\begin{proof}
    A valid commit certificate must contain at least $\lceil \frac{n + f + 1}{2} \rceil$ signatures. 
    Since $\lceil \frac{n + f + 1}{2} \rceil > f$, at least one of these processes must be correct.
    This process had verified the validity of the progress certificate for $x$ in view $u$ before signing.
\end{proof}

\begin{lemma} \label{lem:commitment}
    If there is a valid commit certificate for value $x$ in view $v$,
    then no value other than $x$ can be decided in view $v$.
\end{lemma}
\begin{proof}
    Since any two sets of processes of size $\lceil \frac{n + f + 1}{2} \rceil$ intersect in at least one correct process, no other value can have a commit certificate in the same view.
    Hence, no other value can be decided through the slow path.
    
    Moreover, no other value can be committed through the fast path because any set of size $\lceil \frac{n + f + 1}{2} \rceil$ intersects with any set of size $n-t$ in at least one correct process.
\end{proof}

\begin{corollary} \label{cor:generalized:no-two-decides-in-the-same-view}
    Two different values cannot be decided in the same view.
\end{corollary}
\begin{proof}
    If some value $x$ is decided through the slow path in view $v$, then there must be a commit certificate for $x$ in $v$.
    Hence, by Lemma~\ref{lem:commitment}, no other value can be decided in $v$ (either through the fast or the slow path).

    Two different values cannot be decided through the fast path in the same view because two sets of processes of size $n-t$ intersect in at least one correct process.
\end{proof}

In order to prove consistency, we also need to show that the view change protocol always yields a safe value.
We will do so by induction on the view number.
Suppose that for all view numbers $u < v$, the existence of a progress certificate for value $x$ confirms that $x$ is safe in view $u$.
Now we need to prove that claim for view $v$.

In case when no equivocation is detected, the proof is identical to the one described in Section~\ref{subsec:consistency-proof}.
Otherwise, let us consider the three ways value $x$ can be selected:
\begin{enumerate}
    \item There is a commit certificate for value $x$ in view $w$ in $\votes'$, where $w$ is the highest view in $\votes'$.
    By Lemma~\ref{lem:commit-implies-progress}, there must be a progress certificate for $x$ in view $w$.
    Hence, by the induction hypothesis, $x$ was safe in view $w$ and no value other than $x$ was or will ever be decided in a view $u < w$.
    Moreover, no value $x'$ can be decided in a view $u$ ($w < u < v$) because otherwise $\votes'$ would contain at least one vote for $x'$ in $u$ (which contradict the choice of $w$).
    Finally, by Lemma~\ref{lem:commitment}, no value $x' \neq x$ can be decided in view $w$.
    
    \item There is no commit certificate for any value in view $w$ and there is a set $V \subset \votes'$ of $f+t$ valid votes for a value $x$ in view $w$ in $\votes'$.
    \atrev{By the choice of $w$ and because there is a valid progress certificate for $x$ in $w$, no value other than $x$ can be decided in a view smaller than $w$ or between $w$ and $v$.}
    Moreover, no value can decided in view $w$ through the slow path because otherwise there would be a commit certificate for view $w$ in $\votes'$.
    
    Suppose that some value $x' \neq x$ is decided through the fast path in view $w$. 
    Because we know that $\fleader(w)$ is Byzantine and it does not belong to the set of votes $V$, there are at most $f-1$ Byzantine processes in $V$.
    Hence, the set of $n-t$ processes that acknowledged $x'$ must intersect in at least one correct process with the of $f+t$ processes whose votes are in $V$.
    A correct process cannot acknowledge $x'$ in $w$ and issue a vote for $x \neq x'$ in $w$---a contradiction.

    \item There is no commit certificate for any value in view $w$ and there is no set $V \subset \votes'$ of $f+t$ valid votes for any value in view $w$ in $\votes'$.
    In this case we want to show that any value is safe in $v$.
    \atrev{As in the two other cases}, no value other than $x$ can be decided in a view smaller than $w$ or between $w$ and $v$.
    
    If some value $x$ had been decided in view $w$ through the slow path, there would be a commit certificate for $x$ in view $w$ in $\votes'$ 
    because any set of $\lceil \frac{n+f+1}{2} \rceil$ processes intersects with any set of $n-f$ processes in at least one correct process.
    
    If some value $x$ had been decided in view $w$ through the fast path, there would be at least $f+t$ votes for $x$ in view $w$ in $\votes'$.
    To prove this statement, let us first note that there are votes from at most $f-1$ Byzantine processes in $\votes'$.
    It is sufficient to show that the intersection between the set of $n-t$ processes that acknowledged $x$ and the set of $n-f$ voters in $\votes'$ is at least $(f-1) + (f + t)$ (hence, \atadd{they intersect} in at least $f+t$ correct processes).
    Indeed, given that $n \ge 3f + 2t - 1$, by the pigeonhole principle, any set of $n-f$ processes intersects with any set of $n-t$ processes in at least 
    $(n-f) + (n-t) - n 
    \ge (2f + 2t - 1) + (3f + t - 1) - (3f + 2t - 1)
    \ge (f - 1) + (f + t)$ processes.
\end{enumerate}

\begin{theorem}
    The generalized algorithm satisfies the consistency property of consensus.
\end{theorem}
\begin{proof}
    Suppose, by contradiction, that two processes decided different values $x$ and $y$, in views $v$ and $v'$, respectively.
    Without loss of generality, assume that $v \geq v'$.
    Value $x$ can only be decided in view $v$ if there is a progress certificate for $x$ in $v$ and, as we just proved by induction, such a certificate implies that $x$ is safe in $v$.
    As no value other than $x$ can be decided in a view less than $v$, we have $v=v'$. 
    However, by Corollary~\ref{cor:generalized:no-two-decides-in-the-same-view}, $x = x'$---a contradiction.
\end{proof}

\end{document}